\documentclass{amsart}
\usepackage[utf8]{inputenc}
\usepackage{tikz}
\usepackage{graphicx,hyperref,colortbl,amssymb,adjustbox}
\usepackage{tabularx,booktabs,array,mathdots,nicematrix,tikz}
\usepackage{academicons}
\usepackage{xcolor}

\newtheorem{theorem}{Theorem}
\newtheorem{conjecture}[theorem]{Conjecture}
\newtheorem{corollary}[theorem]{Corollary}

\newtheorem{lemma}[theorem]{Lemma}
\newtheorem{definition}[theorem]{Definition}

\newcommand{\indep}{\perp \!\!\! \perp}

\title{Search and Rescue on a poset}
\thanks{This research was supported by grant N19/007 of the Dutch Ministry of Defence.}
\author{Jan-Tino Brethouwer
\and Robbert Fokkink}
\address{Institute of Applied Mathematics\\
Delft University of Technology\\
Mekelweg 4 \\
2628 CD Delft, The Netherlands\\
ORCiD: 0000-0003-1644-7483\ 0000-0001-8347-4110}
\email{j.f.brethouwer@tudelft.nl\ r.j.fokkink@tudelft.nl} 

\begin{document}

\begin{abstract}
A Search and Rescue game (SR game) is a new type of game on a graph that 
has quickly found applications in scheduling, object detection,
and adaptive search. In this paper, we broaden the definition of SR games
by putting them into the context of ordered sets and Bayesian networks, 
extending known solutions of these games and opening up the way to
further applications.
\end{abstract}
\maketitle

\section{Introduction}

A search-and-rescue game, or SR game, is a two-player win-lose
game that which is played on a \emph{finite} set $X$ of locations.
Hider chooses $h\in X$ and Searcher chooses a permutation
$\sigma$ of $X$, which we call a \emph{search}.
If $\sigma(j)=h$ then $\{\sigma(i)\colon i\leq j\}$ is the
set of
\emph{searched locations}.
For each $x\in X$ there is a Bernoulli random variable $\beta_x$
with known distribution to both players. 
Searcher wins (rescues)
if the product of $\beta_x$ over
the searched locations is equal to one,
and otherwise she loses.
In particular, the search halts as soon as $\beta_x=0$.
Searcher's
payoff $\Pi(\sigma, h)$ is the probability that she wins,
which we call the \emph{rescue probability}. 
Hider wins if he is not rescued, i.e., we have a suicidal Hider.

SR games were first defined by Lidbetter in~\cite{LidbetterSR} and quickly proved to be a fruitful
avenue of new study. 
They have been applied to scheduling problems~\cite{leus, agnetis}, object detection~\cite{balaska}, rendezvous problems~\cite{leone},
and adaptive search~\cite{lidbetter2023}.
For such applications it is natural to impose an order on the set of locations, as 
some jobs can be performed only if other jobs have finished.
In our paper, we therefore consider SR games on a partially ordered set $X$.
A search now has to respect the order and is only allowed if $\sigma(i)<\sigma(j)$ implies $i<j$.
For instance, if $m\in X$ is a maximal  
element for which $x<m$ for all $x\not=m$, then a search halts in $m$. 
We call this an ordered search-and-rescue game, or simply an OSR game.
We also consider the stronger restriction that $\sigma$ is only allowed
if $i<j$ implies $\sigma(i)<\sigma(j)$.
A function with this property is called a  \emph{chain}
~\cite{Stanley} and we say that this game is
a chained SR game, or simply a CSR game.
If we represent the partial order by a Hasse diagram, then we get a
search game on a (directed) graph, which is a topic of its own~\cite{gal}. 
A search is a path on the Hasse diagram in which Searcher may skip
some vertices but is not allowed to backtrack. She can
either search a node or skip it without turning back. 
This appears to be a new game on a graph.

The original SR game as defined by Lidbetter is
a game on an (undirected) graph. 
Searcher chooses a path and Hider chooses a vertex $h$.
Searcher's payoff is the product of all $\beta_v$ over all vertices
$v$ in the path prior to $h$. 
More specifically, let the path be a walk along the vertices
$v_1,v_2,\ldots, v_n$ and let $v_j$ be the first vertex that
is equal to $h$. Then Searcher's payoff is the expected value of
\[
\prod_{x\in\{v_1,\ldots,v_j\}} \beta_x,
\]
i.e., the product of the success probabilities.
Note that $\{v_1,\ldots,v_j\}$ is a multi-set. Searcher is
allowed to backtrack. Even if $x$ occurs twice or more in
the multi-set, then still $\beta_x$ is only sampled once.
If the graph is a rooted tree and if
Searcher can only start from the root, then Hider
selects a leaf because all internal nodes of the tree are dominated.
Lidbetter proves that an optimal search consists of a mix of 
depth-first searches. 
For instance, consider the following example from~\cite{LidbetterSR}
of a rooted tree with three leaves $a,b,c$ and one internal node. The depth
first searches in terms of arrivals at the leaves are $abc, bac, cab, cba$.
\begin{figure}[htbp]
\centering
\begin{tikzpicture}[nodes={draw, circle}, ->]
\node (a) [label=above:{r}] {1/2}
    child  { node (x) [label=above:{i}] {3/5} 
        child { node (b) [label=below:{a}] {1/3} }
        child { node [label=below:{b}] {1/2} }
    }
    child { node [label=below:{c}] {2/3} };
\end{tikzpicture}
\caption{A tree with three leaves $a,b,c$ root $r$ and internal node $i$. Success probabilities displayed in the nodes.}\label{fig:1}
\end{figure}
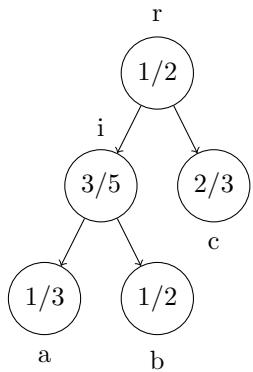
If $h=a$ and the depth-first search is $abc$ then Searcher's payoff is $\frac 1 {10}$. If
the depth-first search is $bac$ then the payoff is $\frac 1{20}$.

In our approach, this SR game on a tree in Fig.~\ref{fig:1} can be defined  
as an OSR game on the unordered set $X=\{a,b,c\}$ with correlated random variables $\beta_a,\beta_b,\beta_c$. The nodes $r,i$ correspond to events that determine
whether or not Searcher rescues. For instance, let $r$ be the event ``no
storm" and let $i$ be the event``high visibility". A rescue in $a$ and $b$
requires no storm and high visibility. A rescue in $c$ only requires
no storm. On top of that, the locations $a,b,c$ each have particular rescue probabilities,
dependent on these events $r,i$. In this approach the tree is a \emph{Bayesian network}~\cite{jensen, pearl}
in which $r,i$ are latent variables.
More formally, let $Pr(x)=Pr(\beta_x=1)$ be the probability of a 
\emph{successful search} of location $x$.
If we put $Pr(a)=\frac 1{10}, Pr(b)=\frac 3{20}, Pr(c)=\frac{1}{3}$ and we put
the conditional 
probabilities $Pr(b|a)=\frac 12, Pr(c|a)=Pr(c|b)=Pr(c|b,a)=\frac 23$,
then we the probability distribution is determined.
We retrieve Lidbetter's game on a tree, in which Searcher can only
search the leaves. The probabilities of the internal nodes correspond
to the correlations between the events $a, b, c$.

What happens to Lidbetter's depth-first search theorem if we replace
the Bernoulli random variables by other random variables? Can we find
effective bounds on the value of the game? What happens if the search
is ordered? These are the
motivating questions for this study.

\medbreak
The paper is organized as follows. We first introduce notation and
terminology. 
We then solve the CSR and OSR games on a poset for uncorrelated Bernoulli
random variables. Suprisingly, despite having a more restricted strategy space,
the CSR game is more difficult to solve. 
We then consider correlated random
variables, extending the work of Lidbetter in~\cite{LidbetterSR} and linking
the SR game to Bayesian networks.
We are able to give bounds on the value of the SR game, but we
can only solve it for simple posets under restrictions on the probability
distribution of the Bernoulli random variables.
In the final section we discuss some further generalizations of SR games.

\section{Definitions and notation}

In this paper we consider games on a partially ordered set (poset) $(X,<)$ with an
associated collection
of Bernoulli random variables $\beta_x$ indexed by $X$. 
An SR game is an example of a \emph{search game} between an immobile hider and a 
mobile searcher on a finite set of locations~\cite[Ch 3]{alperngal}.
The strategy space of player~II (Hider) is the set of locations $X$ and the strategy space
of player I (Searcher) are permutations of subsets of $X$.
Usually, in a search game the payoff is measured in duration or costs and Hider
is the maximizer.
In an SR game, payoff is measured in probability and Hider
is the minimizer.

Throughout the paper we need some basic notions and results 
on posets, for which~\cite{Stanley} is a standard reference.
The set $\{1,\ldots,n\}$ with its standard total order is denoted $[n]$.
A function $f\colon [n]\to X$ is order-reflecting if $f(i)<f(j)$ implies $i<j$.
We consider injective functions only. Locations can be searched only once.
The function is a chain if $i<j$ implies $f(i)<f(j)$, in which case the image $f([n])$
is a totally ordered subset of $X$.
We associate functions with searches and denote them by~$\sigma$ instead of~$f$.

We have a collection of Bernoulli random variables $\beta_x$ with $x\in X$.
These random variables may be \emph{dependent}. 
We write $Pr(S)$ for the probability that $\beta_x=1$ for all $x\in S$. 
If $S$ contains just one or two elements, we write $Pr(x)$ or $Pr(x,y)$.
$Pr(X)$ is the probability that all searches are successful. 

Hider selects a single element $h\in X$.
Searcher selects an order-reflecting function~$\sigma$.
Let $k$ be the minimal number such that $h\in \sigma([k])$. 
If there is no such $k$ then Searcher receives $0$ and Hider receives~$1$.
Otherwise
Searcher's payoff is 
\begin{equation}\label{payoff} 
    \Pi(\sigma,h)=Pr(\sigma([k]))
\end{equation}
\noindent and Hider's payoff is the
complementary probability.
There is an abundance of P's:
$\Pi$ is payoff and $Pr$ is probability.
We say that this is an \emph{OSR game}. 
It models a process of progressive
search. Once $x$ is searched all 
unsearched $y<x$
become inaccessible.
We have a \emph{CSR game} if Searcher is limited to chains. 
It models a process of urgently progressive search.
Once $x$ is searched all $y<x$ and all $y$ 
that are unrelated to $x$ become inaccessible.
We denote the value of an SR game by $V(X)$.
It will be clear from the context what type of SR game we consider.

Both players face a decision problem in which they weigh their options. 
Ratios of probabilities play a role and therefore it is convenient
to consider the \emph{odds} of events. For $x\in X$ we denote
the odds of failure by
\[o_x=\frac{1-Pr(x)}{Pr(x)}.\]
If $Pr(x)=0$ then Hider will hide in $x$ and cannot be rescued. 
We require that $Pr(x)>0$
for all $x\in X$ so the odds are well-defined.
For $S\subset X$ we denote the sum of the odds by
\[
O_S=\sum_{x\in S} o_x.
\]
It
is a quantity that often turns up in decision problems~\cite{bruss}.

Lidbetter proved that an optimal
search and rescue on a tree can be described by weighed coin
tosses,
as in a behavior strategy~\cite{maschler}.
It goes as follows.
Suppose $i$ is an internal node with children $a$ and $b$.
Let $S\subset X$ be the leaves that are descendants of $i$ and let
$A$ and $B$ be the descendants of $a$ and $b$. 
In particular, $S=A\cup B$.
Let $V(A)$ be the value of the subgame on the tree with root $a$,
and let $V(B)$ be the value for $b$. Searcher flips a coin with odds
\begin{equation}\label{flipS}
\frac 1{V(A)}-\frac{Pr(B)}{V(B)}\colon \frac1 {V(B)}-\frac{Pr(A)}{V(A)}
\end{equation}
to decide between $a$ and $b$. 
After reaching a leaf, Searcher continues by backtracking as in
depth-first search.
Note that $V(A)$ is the probability of rescue in subgame $A$, which
is greater than or equal to $Pr(A)$, the probability of 
successful search of all locations.
Therefore $\frac 1{V(A)}\geq 1$ and $\frac{Pr(B)}{V(B)}\leq 1$, so
these odds are well-defined. Also note that the odds favor the subgame with
the smaller rescue probability. 
An optimal search is most likely to start in the most difficult subset first.

Hider can also select $h$ optimally by tossing coins and going down the tree. 
At node $i$ the odds
of choosing between $A$ and $B$ are
\begin{equation}\label{flipH}
\frac {1-Pr(A)}{V(A)}\colon \frac {1-Pr(B)} {V(B)}.
\end{equation}
A suicidal Hider obviously has a preference for the set with the least rescue probability
and the largest probability of failure of the search.

\section{Uncorrelated Search and Rescue}

In this section the random variables $\beta_x$ are uncorrelated.

\begin{lemma}\label{lem:0}
    Suppose $(X,\prec)$ is an extension of $(X,<)$, i.e., $x<y$ implies $x\prec y$.
    Then $V(X,\prec)\leq V(X,<)$ for the OSR game and $V(X,\prec)\geq V(X,<)$ for
    the CSR game.
\end{lemma}
\begin{proof}
If a map is order-reflecting for $\prec$, then it also is order-reflecting
for $<$. In the OSR game, the strategy space for Searcher does not decrease
if we replace $\prec$ by $<$. Therefore, the value of the OSR game does not decrease
if we replace $\prec$ by $<$. If a map is a chain for $<$, then it also is
a chain for $\prec$. By the same argument, the value of the CSR game does not
decrease if we replace $<$ by $\prec$.
\end{proof}

\textbf{Unordered $X$.} Suppose that the partial order on $X$
is trivial, i.e., the only order relation is $x=x$.
Chains are singletons and in the CSR game
Searcher can only search a single location. The payoff
matrix is diagonal with probabilities $Pr(x)$ on the diagonal.
It is optimal for both players to select $x$ with probability 
inversely proportional to
$Pr(x)$. 
The value of the CSR game can be expressed neatly in terms of
the cardinality of $X$ plus the sum of the odds
of failure:
\begin{equation}\label{singleshot}
\frac 1{|X|+O_X}.
\end{equation}

In the OSR game every search is allowed.
This game is one of the motivating examples that
led Lidbetter to define SR games.
It is related to single machine scheduling, see~\cite{LidbetterSR}.
Hider chooses~$x$ with probability
proportional to $o_x$. Searcher starts a search in~$x$
with probability proportional to $o_x$, and continues randomly. 
The value of the OSR game is equal to
\begin{equation}\label{simplesearch}
\frac{1-Pr(X)}{O_X}.    
\end{equation}
In the degenerate case when $Pr(X)=1$ the rescue is
always successful and the odds of failure add up to zero. In this case the value of the
game is equal to 1.

\medbreak
\textbf{Totally ordered $X$.} 
If $X$ is totally ordered, then 
an order-reflecting function is a chain, which means that
the CSR game and the OSR game are the same.
We can identify $X$ with $\{1,\ldots,n\}$.
A search halts as soon as it reaches the greatest element $n$.
Suppose Searcher limits herself to searches in which she
starts in $i$ and searches the remaining locations in increasing order.
Then Hider and Searcher both have $n$ pure strategies and
the payoff matrix $M$ is triangular.
\begin{equation}\label{triangular}
M=
\left[
\begin{array}{rrrrl}
p_1&p_1p_2&p_1p_2p_3&\cdots&p_1p_2p_3\cdots p_n\\
0 & p_2&p_2p_3&\cdots&p_2p_3\cdots p_n\\
0 & 0 & p_3&\cdots&p_3\cdots p_n\\
\vdots&\vdots&\vdots&\ddots&\vdots\\
0&0&0&\cdots&p_n
\end{array}
\right]
\end{equation}
The optimal Searcher strategy in this restricted game
is to start from $i>1$ with probability
proportional to $o_i$ and to start from $1$ with probability
inversely proportional to $p_1$.
The optimal Hider strategy is to choose location $i<n$ with a probability
proportional to $o_i$ and location $n$ with probability 
inversely proportional to $p_n$.
The value of the game with matrix $M$ is equal to
\begin{equation}\label{sumoftheodds}
\frac {1}{1+O_X}.
\end{equation}
Now we limited Searcher's strategies, but notice that Hider's strategy
remains the same if we reorder the locations $1,\ldots,n-1$ in a random way.
Hider's
strategy remains optimal in the full game and therefore we have solved the game.

\begin{lemma}\label{lem:1}
    The value of the OSR game on a poset $(X,<)$ with uncorrelated Bernoulli random variables
    is contained in the interval $\left[\frac{1}{1+O_X},\frac{1-Pr(X)}{O_X}\right].$
    The value of the CSR game is contained in the interval $\left [\frac{1}{|X|+O_X},\frac{1}{1+O_X}\right ].$
\end{lemma}
\noindent Note that $Pr(X)$ is small and $O_X$ is large for large $X$. 
The bounds on the value of the OSR game produce a narrow interval if $X$
is large. We will return to this lemma in the next section for
correlated random variables.

\begin{proof}
    The value of the OSR game on an unordered $X$ is equal to $\frac{1-Pr(X)}{O_X}$, which by
    lemma~\ref{lem:0} puts an upper bound on the value. For the lower bound, consider any extension of $(X,<)$ to a total order. This restricts Searcher and gives the lower bound.
    For the CSR game, the same argument works in the opposite direction by lemma~\ref{lem:0}.
\end{proof}

\textbf{The multi-stage OSR game.}
Let $X=X_1\cup\cdots\cup X_n$ be a union of disjoint sets and put $x<y$ if and only
if $x\in X_i$ and $y\in X_j$ for $i<j$. 
In other words, $X$ is an ordinal sum of unordered subsets,
see~\cite[p 100]{Stanley}.
We say that $X_j$ is \emph{stage} $j$.
Once a search enters a stage it cannot continue in an earlier stage.
The game on this poset is very similar to the game on a total order.
Note that it is optimal to end with a full search of the final stage.
Searching an extra location cannot decrease the rescue probability.

Let $V_j$ be the value of the OSR game restricted to $X_j$,
which is unordered.
We write $P_j=Pr(X_j)$ and $O_j=O_{X_j}$.
By equation~\ref{simplesearch} we have  $V_j=\frac{1-P_{j}}{O_{j}}$.
As in the game on a total order, 
suppose Searcher limits herself to searches of consecutive stages $j, j+1,\ldots,n$.
Within each $X_i$ her search is optimal for the restricted game on this set.
Under this restriction, Searcher essentially only selects the initial stage $j$
and therefore has $n$ strategies.
In response, Hider also selects a stage and applies his optimal strategy, so he essentially 
has $n$ strategies as well. The payoff matrix $L$ is very similar to the
matrix for the total order.
\[
L=
\left[
\begin{array}{rrrrl}
V_1&P_1V_2&P_1P_2V_3&\cdots&P_1P_2P_3\cdots V_n\\
0 & V_2&P_2V_3&\cdots&P_2P_3\cdots V_n\\
0 & 0 & V_3&\cdots&P_3\cdots V_n\\
\vdots&\vdots&\vdots&\ddots&\vdots\\
0&0&0&\cdots&V_n
\end{array}
\right]
\]
The solution of this restricted game is similar to that of the total order. 
Hider chooses stage $j<n$ with probability
proportional to $\frac{1-P_j}{V_j}=O_j$ and stage $n$ with probability 
proportional to $\frac 1{V_n}=\frac {O_n}{1-P_n}$.
Note that $V_n=1$ if $P_n=1$.
Searcher starts in stage $j>1$ with probability
proportional to $\frac{1}{V_j}-\frac{P_{j-1}}{V_{j-1}}$ and starts
in stage $1$ with probability inversely proportional to $V_1$.
In this way, Searcher makes sure that she searches stage $j$
with probability proportional to $\frac 1 {V_j}$.
The value of the game with matrix $L$ is
\begin{equation}\label{multistage}
\frac 1{\frac{P_n O_n}{1-P_n}+O_X}.
\end{equation}
We need to argue that Searcher cannot do better than this.
Suppose we combine $X_{j-1}\cup X_j$ for $j<n$, 
then the number of stages goes down by one which gives more
freedom to Searcher, but we argue that
Hider's optimal strategy does not change.
In the game with $n$ stages, Hider selects stage $j<n$
with probability proportional to $O_j$. 
He then selects $x\in X_j$ with probability proportional
proportional to $o_x$.
This is the same as first selecting stage $X_{j-1}\cup X_j$
with probability proportional to $O_{j-1}+O_j$
and then tossing a coin with odds
$O_{j-1}\colon O_j$ to select from $X_j$.
We can combine or split up stages without affecting
the optimal Hider strategy. The game is
equivalent to the two-stage game on
the ordinal sum $(X\setminus X_n)\cup X_n$.
All searches that start in $X\setminus X_n$ and end
with a full search of $X_n$ have expected probability
of rescue given by equation~\ref{multistage}.
\medbreak
An element $m\in X$ is a \emph{maximum} if there are no $x>m$. 
Let $M\subset X$ be the subset of maxima. An optimal search
always ends in $M$. If $\sigma$ does not end in $M$, then it
can be extended by some $m\in M$ without decreasing the
probability of rescue. Since there is no order relation between
different elements of $M$, an optimal search ends with a full
search of $M$.

\begin{theorem}\label{thm:1}
    Let $M\subset X$ be the subset of maxima. Let $P_M$ be the product of all success
    probabilities over $m\in M$ and let $O_M$ be the sum of the odds over $M$.
    The value of the OSR game
    on $X$ is equal to
    \begin{equation}\label{OSRvalue}
    \frac 1{\frac{P_M O_M}{1-P_M}+O_X}.
    \end{equation}
\end{theorem}

\begin{proof}
    Let $(X\setminus M) \cup M$ be the ordinal sum of two unordered subsets. 
    The value
    of the OSR game on this poset is given by equation~\ref{OSRvalue}. 
    By lemma~\ref{lem:0} this puts an upper bound on the value of the OSR
    game on $X$. Conversely, extend the partial order of $X\setminus M$ to a total order
    and consider the OSR game on the resulting poset. By lemma~\ref{lem:0} the value
    of the game on this poset puts a lower bound. This is a multi-stage game that ends in 
    $M$ such that all previous stages are singletons.
    The value of this game is again given by equation~\ref{OSRvalue} and the lower
    bound equals the upper bound.
\end{proof}

This theorem settles OSR games and we now turn our attention to
CSR games. 
Our next lemmas sharpen the bounds of lemma~\ref{lem:1}.

\begin{lemma}\label{lem:2}
    Let $M\subset X$ be the subset of maxima. If Hider hides in $m\in M$ with probability
    inversely proportional to $Pr(m)$ and in $x\in X\setminus M$ with probability proportional
    to $o_x$, then the rescue probability is at most \[ \frac{1}{|M|+O_X}\] for
    any search. This puts an upper bound on the value of the game.
\end{lemma}
\begin{proof}
    For any chain $x_1, x_2,\ldots, x_k$ that ends in $x_k\in M$, the expected rescue probability is
    \[
    \frac{(1-Pr(x_1))+Pr(x_1)(1-Pr(x_2))+\cdots + Pr(x_1)\cdots Pr(x_{k-1})}{|M|+O_X}
    =
    \frac{1}{|M|+O_X}.
    \]
    If the chain does not end in a maximum, then the rescue probability is less.
\end{proof}

An antichain is a subset $S\subset X$ such that any two elements of $S$ are incomparable.
The \emph{width} of $X$ is the size of its largest antichain. By Dilworth's theorem~\cite{dilworth}, it 
equals the minimum number of disjoint chains into which the set can be partitioned

\begin{lemma}\label{lem:3}
    Let $w$ be the width of $X$. There exists a mixed Searcher
    strategy with rescue probability at least \[\frac{1}{w+O_{X}}.\]
    This puts a lower bound on the value of the game.
\end{lemma}

\begin{proof}
Consider the following strategy for Searcher.
She decomposes $X$ into $w$ disjoint chains and selects a
random $x\in X$ to start the search.
The search proceeds along the unique chain that contains it. 
Searcher selects initial elements $x$ with probability inversely proportional
to $Pr(x)$ and all other elements with probability proportional to $o_x$.
Suppose Hider selects $h$ and let $S=\{x_1,\ldots,x_k\}$ be the
chain that contains $h$ with $x_j=h$.
Within this subset $S$ we have the game on a total order in which
Searcher applies the optimal strategy. The rescue probability
is the same for all elements. Hider is indifferent between all elements
of $X$.
\end{proof}

In particular, we have solved the game if $M$ is a maximal antichain of $X$.
For instance, if the Hasse diagram is a tree. The upper bound on the game 
depends on $|M|$ which can be computed in linear time by depth-first search. 
The lower bound on the game depends on $w$ which can be computed in polynomial time~\cite{felsner}.

\medbreak
\textbf{The multi-stage CSR game.} 
We reconsider the ordinal sum $X=X_1\cup\cdots\cup X_n$ of unordered subsets,
this time for the CSR game, in which Searcher can only search one location
from each $X_i$. In the OSR game, Searcher can search all locations from $X_i$.
As before, we write $O_j$ for the sum of the odds over the elements in the
stage $X_j$ for $i$ and we write $O_{\leq j}$ for the sum of the odds over
all stages $\leq j$. 
Let $k$ be such that $|X_k|+O_{\leq k}$ is \emph{maximal}.
By lemma~\ref{lem:2} Hider can limit the
rescue probability to 
\[\frac{1}{|X_k|+O_{\leq k}}\] by choosing $h$ from
the first $j$ stages only. We show that Searcher can actually 
achieve this rescue probability.

We add an initial node $s$ at stage zero and a terminal node $t$ 
at stage $n+1$ with $Pr(s)=Pr(t)=1$. This does not change the game
as $s$ and $t$ can be searched for free and therefore Hider will not select
these nodes.
We represent the search strategy as a network flow on
the Hasse diagram of this order, which
has a single source $s$ and a single sink $t$.
Searcher starts
from $s$ with probability one. 
We need to prove that Searcher achieves a rescue probability
$\geq \frac{1}{|X_k|+O_{\leq k}}$ at each node.
Instead of probability one, we give Searcher a total weight of $|X_k|+O_{\leq k}$ 
at node~$s$. 
Now she needs to make sure that the rescue probability is $\geq 1$
for each node.
The inflow $I_x$ at each node $x$ has to be $\geq \frac{1}{Pr(x)}$ which is
needed to achieve rescue probability one. 
Excess flow can skip the node.
The rescue fails with probability $1-Pr(x)$ and therefore takes away $o_x$
from the flow.
The outflow is $I_x-o_x$. The entire stage $j$ takes away $O_j$ from the flow.
After $j<k$ stages the total flow is reduced to $O_{\leq k}-O_{\leq j}+|X_k|$.
The next stage requires an inflow of $O_{j+1}+|X_{j+1}|$ which can be satisfied
since
\[
O_{\leq k}-O_{\leq j}+|X_k|\geq O_{j+1}+|X_{j+1}|
\]
by the fact that $O_{\leq i}+|X_i|$ is maximal at $k$.
After $k$ stages, the outflow is equal to $|X_k|$.
The required inflow for a stage $j>k$ is equal to $O_j+|X_j|$ and this
stage takes away $O_j$ from the flow.
Since we have 
\[
|X_k|-O_{k+1}-\ldots -O_j\geq |X_j|
\]
the required inflow can always be met.
This guarantees a rescue probability $\geq 1$ for each node and if we normalize
the probability than we find that the value of the game is
$\frac{1}{O_k+|X_k|}$.
\medbreak
An antichain $A\subset X$ is \emph{maximal} if for every $x\in X$ there is an $a\in A$ such that
either $a\leq x$ or $x\leq a$. A maximal antichain partitions the poset $X$ into 
$A^{-}=\{x\colon \exists a\in A, x\leq a\}$ and $A^+=\{x\colon \exists a\in A, a<x\}$.
In other words, $A$ is a cut.
The elements of $A$ are maximal in the poset $A^{-}$ with the order inherited from $X$.
Hider can guarantee a rescue probability $\leq O_{A^-}+|A|$ by hiding in $x\in A^{-}$
with probability $o_x$ if $x\not\in A$ and probability inversely proportional to $Pr(x)$
if $x\in A$.

\begin{theorem}\label{thm:2}
Let $A$ be a maximal antichain that maximizes $O_{A^{-}}+|A|$.
The value of the CSR game on $X$ is equal to
\[
\frac{1}{O_{A^{-}}+|A|}.
\]
\end{theorem}
\begin{proof}
Hider can achieve this rescue probability. We need to prove that Searcher can achieve it as well. 
We add a source $s$ and a sink $t$ with rescue probability one,
with $s<x<t$ for all $x\in X$. Searcher's task is to find a network flow such that the
inflow $I_x\geq \frac 1{Pr(x)}$ for all nodes. Each node reduces the flow by $o_x$.
We have a network flow problem but unlike the standard max-flow min-cut problem 
it is max-cut min-flow and the flow is
dissipative. 
However, it is possible to modify the proof of the
max-flow min-cut theorem for this problem.
We have adapted the proof given by Trevisan \cite[Ch 15]{trevisan}.

Let $E$ be the edge set and let $f_{xy}$ be the flow on $(x,y)\in E$.
There are two constraints at each $x\in X$: inflow at least
$\frac{1}{Pr(x)}=o_x+1$ to guarantee rescue probability one,
and dissipation (inflow minus outflow) 
is at least $o_x$ to guarantee a full search of the node. 
\begin{equation*}
\begin{array}{ll@{}ll}
\text{minimize}  & \displaystyle\sum\limits_{(s,x)\in E} f_{sx} &\\
\text{subject to }\forall x\in X\\
                 & \displaystyle\sum_{(v,x)\in E} f_{vx} &\ \geq o_x+1,\\
                               & \displaystyle\sum_{(v,x)\in E}  f_{vx} - \displaystyle\sum\limits_{(x,w)\in E} f_{xw}&\ \geq o_x,\\
\end{array}
\end{equation*}
where all $f_{xy}\geq 0$.
The dual problem has two variables for the two constraints at each $x\in X$
and has one constraint for each $(v,x)\in E$.
\begin{equation*}
\begin{array}{ll@{}ll}
\text{maximize}  & \displaystyle\sum_{x\in X} (o_x+1)g_x + o_x h_x, &\\
\text{subject to }\\
                 & g_x + h_x\leq 1,&\forall (s,x)\in E,\\
                 & g_x + h_x - h_v\leq 0,&\forall (v,x)\in E,\ v\not=s,\\
\end{array}
\end{equation*}
where all $g_x,h_x\geq 0$. If we put $g_a=1$ for all $a\in A$ a 
maximal antichain and $h_b=1$ for all $b<A$, and zero elsewhere,
then the constraints of the dual problem are satisfied.
This choice is feasible and
in this case \[\sum\limits_{x\in X} (o_x+1)g_x + o_x h_x=O_{A^-}+|A|.\]
We need to show that this is a solution of the dual problem for some antichain,
because then the minimax theorem implies that Searcher has a feasible flow that
achieves the required rescue probability.

Let $g_x$ and $h_x$ be a feasible solution of the dual problem.
Pick $0\leq T\leq 1$ uniformly at random and put $\bar h_x=1$
if $h_x\geq T$ and $\bar g_x=1$ if $h_x+g_x\geq T> h_x$.
For all other $x$ we
set $\bar h_x=\bar g_x=0$. 
This implies that $\bar h_x+\bar g_x\leq 1$ and if
$\bar h_x+\bar g_x=1$ then $\bar h_v=1$ for all $(v,x)\in E$.
Therefore $\bar g_x$ and $\bar h_x$ are feasible.
We claim that $A=\{x\colon \bar g_x=1\}$
is a (random) antichain. In other words, there is no path between
any $a,b\in A$. To see this, notice that $h_a<T$ and by the
second constraint $g_x+h_x<T$ for all successors of $a$.
If $g_b+h_b\geq T$ then $b$ is not a successor of $a$ and
$A$ is indeed an antichain. 

The value of
$
\displaystyle\sum_{x\in X} (o_x+1)\bar g_x + o_x \bar h_x
$
is equal to $O_{A^-}+|A|$, since all $x\in A^{-}$ either have $\bar g_x=1$ or
$\bar h_x=1$ and all $x\in A$ have $\bar g_x=1$.
The expected value is equal to
\[
\displaystyle\sum_{x\in X} (o_x+1) Pr(g_x + h_x \geq T > h_x) + \sum_{x\in X} Pr(h_x \geq T)
= \displaystyle\sum_{x\in X} (o_x+1) g_x + \sum_{x\in X} h_x,
\]
which is maximal. 
Therefore, there must be an antichain such that $O_{A^-}+|A|$ equals the solution
of the dual problem, as required.
\end{proof}

The optimal Searcher strategy is a feasible solution of an LP problem and
can therefore be computed in polynomial time in terms of $|X|$. 
The optimal Hider strategy
can also be computed in polynomial time, by solving the dual problem and
random rounding.

\section{Correlated Search and Rescue}

In this section the random variables $\beta_x$ are correlated.
This makes the games much more difficult to solve. We are only able to
solve the games for simple posets under strong restrictions on
the distribution of the Bernoulli random variables.

We first consider the OSR game on an unordered $X$:
any permutation of $X$ is an admissible search.
If there
are only two locations $X=\{a,b\}$ we have a simple symmetric matrix game
\[
\begin{pmatrix}
    Pr(a)&Pr(a,b)\\
    Pr(a,b)&Pr(b)
\end{pmatrix},
\]
which is easy to solve. For three locations the game is already much more elaborate.

\begin{definition}\label{def:1}
    For three locations $x,y,z$ we say that $x$ and $y$ are 
    {\em{conditionally independent}} with respect to $z$
    if  \[Pr(x,y|z)=Pr(x|z)Pr(y|z).\]
    It is denoted $(x\indep y|z)$.
\end{definition}

The common way to define conditional independence is for events,
which in our case are successful
searches of locations. 
We note that if $A,B,C$ are three events such that $A\indep B|C$, then
it does not necessarily hold that $A\indep B|\overline C$, where $\overline C$
denotes the complementary event of $C$.

Conditional independence is often interpreted in terms of
learning~\cite{pearl}.
An equivalent way to define $(x\indep y|z)$ is $Pr(x|z,y)=Pr(x|z)$. 
In terms of learning:
if we learn that $z$ has happened, then the probability of
$x$ gets a Bayesian update $Pr(x|z)$, but if
we then learn that $y$ has also happened, we learn nothing new.
In the SR game on three locations $\{a,b,c\}$
in Fig.~\ref{fig:1} we have that $Pr(c|a)=Pr(c|b)=Pr(c|b,a)$ 
and therefore both $a\indep c|b$ and
$b\indep c|a$.
\medbreak
\textbf{The OSR game on three unordered locations.} 
Let $X=\{a,b,c\}$ and 
suppose that both $a\indep c|b$ and
$b\indep c|a$. We build a graphical model
of this probability distribution in Fig.~\ref{fig:2}
in analogy of the SR game on a tree in Fig.~\ref{fig:1}.

\begin{figure}[h]
\centering
\tikzset {
    my_circle/.style = {circle, draw= black}
}
\begin{tikzpicture}[->, sibling distance=2cm, level distance = 1.5cm]
\node (a) {}
    child {node [my_circle, fill={gray!50}] {q}
    child { node [my_circle, fill={gray!50}]{r} 
        child { node [my_circle](b) {a} edge from parent node[left] {$Pr(a|b)$}}
        child { node [my_circle]{b} edge from parent node[right] {$Pr(b|a)$} }
        edge from parent node[left] {$\frac{Pr(b|c)}{Pr(b|a)}$}
    }
    child { node [my_circle]{c} edge from parent node[right] {$Pr(c|b)$}}
     edge from parent node[left] {$\frac{Pr(c)}{Pr(c|b)}$}};
\end{tikzpicture}
\caption{The tree of Fig.~\ref{fig:1} revisited as a network with events $a,b,c$
and internal nodes $q,r$.
The events $a,b,c$ represent successful rescue in these three locations. 
The weights of the directed edges are ratios of conditional probabilities. 
}\label{fig:2}
\end{figure}
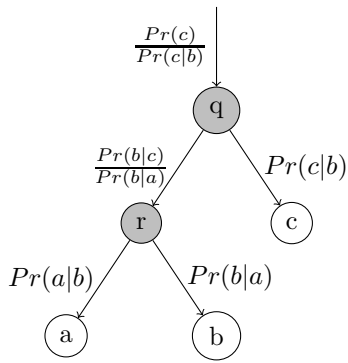
The three leaves of the graphical model are labeled by
the three locations. Contrary to Fig~\ref{fig:1},
the internal nodes $q,r$ do not
necessarily correspond to events. 
The edge weights are \emph{ratios}
of conditional probabilities, which may be larger than one,
which is why we say this is a \emph{pseudo-Bayesian network}.
In Bayesian networks, nodes correspond to
events and edge weights are conditional probabilities, which are at most one.
To compute the rescue probability of a leaf, multiply the edge weights
of the path to that leaf. More generally, to compute
$Pr(S)$ for $S\subset \{a,b,c\}$ multiply the
edge weights of its spanning tree. For instance,
\[
Pr(a,c)=\frac{Pr(c)}{Pr(c|b)}\cdot\frac{Pr(b|c)}{Pr(b|a)}\cdot 
Pr(c|b)\cdot Pr(a|b)=\frac{Pr(b,c)\cdot Pr(a)}{Pr(b)}
\]
which is equal to $Pr(c|b)\cdot Pr(a)=Pr(c|a)\cdot Pr(a)$, for this tree. This is
indeed equal to $Pr(a,c)$. 

In this setup, searching a leaf requires the probabilities
of the path leading up to it, as in Lidbetter's SR game.
If all edge weights are bounded by one, 
then we are back at the original SR game on a tree. 
In fact, there is no need to bound the weight 
$\frac{Pr(c)}{Pr(c|b)}$
of
the first incoming edge.
It is a common
divisor of all rescue probabilities. 
It is a sunk cost that can be replace by one,
without changing the optimal strategies of the players.
then we multiply all payoffs by $\frac{Pr(c|b)}{Pr(c)}$ and the game
remains essentially the same. 

Suppose Searcher first successfully searches $b$.
The probabilities of the two remaining locations are now
updated to $Pr(a|b)$ and $Pr(c|b)$.
The Bayesian factors for these updates satisfy
\[
\frac{Pr(c|b)}{Pr(c)}\leq\frac{Pr(a|b)}{Pr(a)}
\]
because the paths of $a$ and $b$ have more edges
in common than the paths of $c$ and $b$.
The weights of these edges are taken out by the update,
which is why the Bayes factor for $a$ is larger and
hence the incentive to now search $a$ is stronger. This is a Bayesian interpretation of
depth-first search.
\medbreak
In order to extend this example from three to more locations, we need some
further terminology. First of all, conditional independence extends from
three events as in definition~\ref{def:1} to collections of random variables.
It is a central topic of study in probabilistic learning~\cite{studeny}. 
We will need a rather special notion that is tailored to SR games.
\begin{definition} [co-independence]
    Let $X=A\cup B$ be a partition into disjoint sets such that
    for any $A'\subset A, B'\subset B$ and any $a\in A, b\in B$ it holds that
    $Pr(A'|B')=Pr(A'|b)$ and $Pr(B'|A')=Pr(B'|a)$. Furthermore, 
    $Pr(A'|a,b)=Pr(A'|a)$ and $Pr(B'|a,b)=Pr(B'|b)$. Then we say that $A$
    and $B$ are \emph{co-independent}, denoted $X=A\Vert B$ or simply $A\Vert B$.
\end{definition}
Informally
speaking, each element of $A$ teaches us the same about $B$ and any element from $B$
teaches us more about $B$ than all of $A$.  

\begin{definition}
    Let $X$ be a set of two locations or more.
    We say that $X$ is \emph{reducible} if it admits a partition $X=A\cup B$
    into disjoint
    non-empty subsets $A\Vert B$. We say that it is \emph{completely reducible} if 
    every subset $Y\subset X$ containing two elements or more is reducible.
\end{definition}
In particular, if $X$ is completely reducible then $X=A_0\Vert A_1$, which can be partitioned
into $A_0=A_{00}\Vert A_{01}$ and $A_1=A_{10}\Vert A_{11}$ etc.
In the SR game on a tree, the probability distribution on the leaves is
completely reducible.
For instance, in the SR game of Fig.~\ref{fig:1} we have $\{a,b\}\Vert{c}$
and $\{a\}\Vert\{b\}$.

\begin{lemma}\label{Bfactor}
    If $X=A\Vert B$ then the Bayes factor
    \[
    \frac{Pr(A'|B')}{Pr(A')}
    \]
    is the same for any nonempty $B'\subset B$ and $A'\subset A$.
\end{lemma}
\begin{proof}
    Let $a\in A$ and $b\in B$. By co-independence and the properties of conditional probability
    \[\frac{Pr(A'|B')}{Pr(A')}=\frac{Pr(A'|b)}{Pr(A')}=\frac{Pr(b|A')}{Pr(b)}=\frac{Pr(b|a)}{Pr(b)},\]
    which does not depend on $A'$ and $B'$. 
\end{proof}

\begin{theorem}
    Suppose that $X$ is completely reducible such that $X=A_0\Vert A_1$ 
    and each subsequent $A_w=A_{w0}\Vert A_{w1}$ for a binary word $w$.
    Then the probability distribution can be represented
    by a pseudo-Bayesian tree with nodes $A_w$ and root $X$, which has an incoming edge
    without an initial node.
    Each $A_w$ has sibling $A_{\bar w}$, such that $w$ and $\bar w$ have common prefix $v$ and only differ in the last digit. 
    In particular, $A_v$ is the parent of $A_w$. The incoming edge of $A_w$ has weight
    \[
    \frac{Pr(a_{w0}|a_{\bar w})}{Pr(a_{w0}|a_{w1})},
    \]
    for arbitrary elements $a_w\in A_w$.
    The incoming edge of the root has weight $\frac{Pr(a_0)}{Pr(a_0|a_1)}$, because
    the root has no sibling.
    If
    $T$ is the spanning tree of $S\subset X$ then $Pr(S)$ is the product of the
    weights of the edges in $T$. 
\end{theorem}

\begin{proof}
By induction, the conditional distribution $Pr(\cdot | a_i)$ on $A_i$ for $i\in\{0,1\}$
can be represented by pseudo-Bayesian trees. 
These are the two subtrees that spring from the root $X$.
Let $S_i=S\cap A_i$ for $i\in\{0,1\}$. Then
\[
Pr(S)=Pr(S_1|S_0)Pr(S_0)=Pr(S_1|a_0)Pr(S_0).
\]
By induction $Pr(S_1|a_0)$ equals the product of the weights for
its spanning tree $T_1$ in the pseudo-Bayesian subtree for $A_1$.
By lemma~\ref{Bfactor}
\[
Pr(S_0)=Pr(S_0|a_1)\cdot\frac{Pr(a_0)}{Pr(a_0|a_1)}
\]
the right-hand side of which is equal to the product of the
weights over the spanning tree $T_0$ and the weight
of the incoming edge at $X$.
\end{proof}
Note that we may switch the labels $0$ and $1$
and therefore
$    \frac{Pr(a_{w0}|a_{\bar w})}{Pr(a_{w0}|a_{w1})}=
    \frac{Pr(a_{w1}|a_{\bar w})}{Pr(a_{w1}|a_{w0})}$.
    Indeed, this is a consequence of co-independence, since
\[\frac{Pr(a_{w0}|a_{\bar w})}{Pr(a_{w1}|a_{\bar w})}=
    \frac{Pr(a_{\bar w}|a_{w0})Pr(a_{w0})}{Pr(a_{\bar w}|a_{w1})Pr(a_{w1})}=
    \frac{Pr(a_{w0})}{Pr(a_{w1})}=\frac{Pr(a_{w0}|a_{w1})}{Pr(a_{w1}|a_{w0})}.\]
    
    If all weights are bounded by one (with the possible exception
    of the incoming edge of the root, which is a sunk cost),
    then we have a proper Bayesian network of events.
    In this case, the game is equivalent to an
    SR game on a tree.

\begin{corollary}
    Lidbetter's solution of the SR game on a tree extends to the OSR game for
    a completely reducible distribution on an unordered $X$ provided
    that all Bayesian factors $\frac{Pr(a_{w0}|a_{\bar w})}{Pr(a_{w0}|a_{w1})}\leq 1$.
\end{corollary}

\begin{corollary}
     Let $X$ be unordered and completely reducible such that all weights on the pseudo-Bayesian network are $\leq 1$ including the weight on the incoming edge at the root. The OSR game on $X$ has value $V(X)\leq \frac{1-Pr(X)}{O_X}$.
\end{corollary}

\begin{proof}
    Let the decomposition of $X$ begin with $A\Vert B$ and let $w$ be the weight of the
    incoming edge at the root. Searcher either starts with an exhaustive
    search of $A$ followed by $B$, or vice versa. Hider hides in $A$ or in $B$.
    Essentially we have a $2\times 2$ payoff
    matrix
    \[
    \begin{pmatrix}
        wV(A)&wPr(A)V(B)\\
        wPr(B)V(A)&wV(B)
    \end{pmatrix},
    \]
    where $V(A), V(B)$ are the values of the subgames on the two trees
    that spring from the root and $Pr(A), Pr(B)$ are the products of
    the weights of these two trees.
    Searcher incurs the cost $w$ and flips a coin as in equation~\ref{flipS}. The value of
    the game is 
    \[
    \frac{w(1-Pr(A)Pr(B))}{\frac{1-Pr(A)}{V(A)}+\frac{1-Pr(B)}{V(B)}}.
    \]
    By induction the denominator is $\geq O_A+O_B=O_X$.
    We also have $Pr(X)=w\cdot Pr(A)\cdot Pr(B)$.
    Therefore, the denominator is $w-Pr(X)\leq 1-Pr(X)$ by our assumption
    that $w\leq 1$.
\end{proof}

Lidbetter showed that depth-first search is optimal. 
Searcher
starts at the root and continues until it reaches a leaf 
and backtracks, continuing at an unsearched node
that is closes to the leaf.
What can be said if all Bayesian factors are $\geq 1$?
A search is \emph{backjumping}~\cite{vanbeek} if it 
continues at an unsearched node that is closest to the root.
It is a common method in AI and automatic theorem proving~\cite{restart}
to prevent a search more dispersive.
To specify this, suppose that a node $x$ has two children $a$ and $b$ with
offspring $A$ and $B$. If a search first visits a leaf in $A$ then it does not re-enter
$A$ before visiting a leaf of $B$. 
Backjumping search is the opposite of backtracking search
It is a natural procedure
if all Bayesian factors are $\geq 1$, when it is unlikely that
a search of a nearby leaf will be successful.
Numerical experiments show that it is not true that the optimal search is a mix of 
backjumping searches if all Bayesian factors are $\geq 1$.
They do indicate that the following weaker statement is true.

\begin{conjecture}\label{conj:1}
    Consider the OSR game with a completely reducible distribution
    with all Bayesian factors~$\geq 1$. 
    An optimal Searcher strategy contains a pure strategy
    that is backjumping.
\end{conjecture}

\begin{definition}
    We say that the probability distribution on $X$ is \emph{positively
    correlated} if $Pr(A|B)\geq Pr(A)$ for all $A,B\subset X$.
    We say that it is \emph{negatively correlated} if $Pr(A|B)\leq Pr(A)$.
\end{definition}

This property extends to arbitrary random variables for which it
is known as positive or negative
\emph{association}~\cite{assoc}.
Note that a completely reducible distribution on $X$ with all weights $\leq 1$
is positively correlated.
It is negatively correlated if all weights are $\geq 1$.
For positively correlated $X$ the search gets progressively easier, which is
a general assumption in search games~\cite{set}.
For negatively correlated $X$ it gets progressively more difficult, which
may be a more natural assumption for a rescue operation.

\begin{lemma}\label{lem:3}
    The value of the OSR game on a positively correlated poset $X$ 
    is $\geq \frac{1}{1+O_X}.$ 
    The value of the CSR game on a negatively correlated poset
    $X$ is contained in the interval $[\frac{1}{|X|+O_X},\frac{1}{1+O_X}].$
\end{lemma}
\noindent In particular, the value of the OSR game on a completely reducible $X$
with all weights $\leq 1$ is contained in the interval $[\frac{1}{1+O_X},\frac{1-Pr(X)}{O_X}]$.
\begin{proof}
For a positively correlated distribution we have that 
\[Pr(A)=Pr(A\setminus \{a\}|a)\cdot Pr(a)\geq Pr(A\setminus\{a\})P(a)\]
for any $a\in A.$
By induction, $Pr(A\setminus\{a\})\geq \prod_{b\in A\setminus\{a\}} Pr(b)$.
Therefore, the payoff matrix is bounded from below by the matrix for the
uncorrelated SR game. By lemma~\ref{lem:1} we conclude that
$V(X)\geq \frac{1}{1+O_X}$ for the OSR game.
For a negatively correlated distribution we have the opposite signs and
we conclude that $V(X)\leq \frac{1}{1+O_X}$ for the CSR game.
The worst case CSR game is an unordered $X$ in which Searcher can only
search one location. The payoff matrix is diagonal with entries $Pr(a)$.
This game has value $\frac{1}{|X|+O_X}$.
\end{proof}

So far we mainly considered unordered $X$, which is the least restrictive
for Searcher in the OSR game. A total order is most restrictive and we
are only able to solve that game in a limited case.
Before doing that, we first exhibit some examples to show  
that the solution of the OSR game is not trivial.

\medbreak

\textbf{The OSR game on a total order.}
Let $X=\{a,b,c\}$ be a total order $a<b<c$. The payoff matrix
of the game is
\[
\begin{pmatrix}
    Pr(a)&Pr(a,b)&Pr(a,b,c)\\
    Pr(a)&0&Pr(a,c)\\
    0&Pr(b)&Pr(b,c)\\
    0&0&Pr(c)
\end{pmatrix}.
\]
We write $\bar x$ for the event that the search in $x$ is unsuccessful.
A straightforward but tedious computation shows that the solution
of the game depends on the Bayes factor
\begin{equation}\label{factor}
\frac{Pr(a|\bar b, c)}{Pr(a|b)}.
\end{equation}

If it is more than one, Searcher's optimal strategy mixes the searches 
in rows $2,3,4$. If it is less than one, she mixes the searches in rows $1,3,4$.

\medbreak

\textbf{The CSR game on a star.} Let $X=U\cup \{*\}$ where $U$ is unordered
and $u<*$ for all $u\in U$. The Hasse diagram of this poset is a star graph.
A search either consists of $*$ or of $\{u,*\}$. Suppose that $*$ is 
independent of $U$, i.e., $Pr(u,*)=Pr(u)\cdot Pr(*)$. 
If Searcher selects
$\{u,*\}$ with probability proportional to $\frac 1{Pr(u)}$, then Hider
is indifferent between all locations in $U$. The rescue probability in $*$
is proportional to $|U|\cdot Pr(*)$. If $Pr(*)>\frac 1{|U|}$ then Hider
does not hide in $*$. The game reduces to the CSR game on $U$ with value
$\frac 1{O_U+|U|}$. Not all Hider strategies are active.

If $Pr(*)<\frac {1}{|U|}$ then Searcher includes 
the search $\{*\}$ with probability proportional to $\frac1{Pr(*)}-|U|$
to make Hider indifferent. Under this strategy the rescue probability is
$\frac 1{O_X+1}$ for all locations.
If Hider hides in $u$ with probability
proportional to $\frac {1-Pr(u)}{Pr(u)}$ and in $*$ with probability
proportional to $\frac 1{Pr(*)}$,
then the rescue probability is $\frac 1{O_X+1}$.
This is the value of the game.

The CSR game on a star falls apart in two separate cases.
Now consider an OSR game on a total order
$X=\{1,2,\ldots, n\}$.
such that $Pr(A)=0$ for all sets of cardinality $>1$ except for sets $A=\{j,n\}$.
This is the CSR game on a star, again illustrating that the OSR game on a total order
falls apart into different cases.

\begin{theorem}
    Let $X=\{1,2,\ldots,n\}$ be a negatively correlated total order such that $n$ is independent
    of the other locations, i.e., $Pr(A|n)=Pr(A)$ for all $A\subset X\setminus\{n\}$, and
    such that $Pr(n)\leq \frac 1{n-1}$.
    Then the value of the game is $V(X)=\frac 1{O_X+1}$.
\end{theorem}

\begin{proof}
    By lemma~\ref{lem:3} we know that $V(X)\leq \frac 1{O_X+1}$. It suffices to
    find a Searcher strategy that guarantees this rescue probability.
    Suppose Searcher limits herself to a search of only one or two locations,
    one of which is $n$. Then we have the SR game on a star with uncorrelated
    random variables, which by the
    example above has the required value.
\end{proof}

\section{Further generalizations of SR games}

There
are several ways to generalize SR games that may deserve
further study. We discuss some
of them here, without going into a full analysis. 

\medbreak
\textbf{General random variables.}
The most straightforward generalization is to replace
the $\{0,1\}$-valued random variables by arbitrary non-negative 
random variables.
Let's denote the set of locations by $[n]=\{1,\ldots,n\}$ with partial
order $\prec$
so that we can
denote the random variable of location $i$ by $X_i$. 
A possible interpretation of $X_i$ is that the rescue comes with
a certain (random) reward that depends on the location.
For example, an adversarial Hider that is out to do damage
can be more or less detrimental, depending on the location.
In this setting it is natural to replace $\{0,1\}$-valued variables
by weighted variables.
In analogy of equation~\ref{payoff}
we now have a zero-sum game with payoff equal to the expected value 
of the product over the visited locations:
\[\Pi(\sigma,h)=\mathbf E[X_{i_1}\cdot X_{i_2}\cdots X_{i_k}],\]
where $i_1,i_2,\ldots,i_{k-1}$ are the searched locations before 
arriving at the hideout $h=i_k$. 
If the random variables are independent and the expected values
are $\leq 1$, then the payoff matrix for this game is equivalent
to a payoff matrix for Bernoulli random variables. Our results
from section 2 carry over to this case.

If the expected values are not bounded by one, then the game gets
more difficult to solve. We illustrate that for the OSR game on
a total order $[n]$ with expected values $e_i=E[X_i]$. 
If we limit Searcher to searches of consecutive locations, then
we get the payoff matrix of equation~\ref{triangular}:
\[
A=
\left[
\begin{array}{rrrrl}
e_1&e_1e_2&e_1e_2e_3&\cdots&e_1e_2e_3\cdots e_n\\
0 & e_2&e_2e_3&\cdots&e_2e_3\cdots e_n\\
0 & 0 & e_3&\cdots&e_3\cdots e_n\\
\vdots&\vdots&\vdots&\ddots&\vdots\\
0&0&0&\cdots&e_n
\end{array}
\right]
\]
We say that
a subset of consecutive locations $i,i+1,\ldots,j$ is a \emph{run}
if all cumulative products $e_i, e_ie_{i+1}, \ldots, e_ie_{i+1}\cdots e_j$
are $\geq 1$. For a run we have that row $i$ of the matrix dominates
the following rows from $i+1$ up to $j+1$. We can delete these rows
from the matrix. 
If $j$ happens to be the final location $n$, then we only delete rows
up to $j$.
Columns $i$ up to $j+1$ (or $n$) in the remaining matrix
are multiples of column $i$ by factors $1, e_{i+1}, e_{i+1}e_{i+2},\ldots,
e_{i+1}e_{i+2}\cdots e_{j+1}$. 
The column with the minimum multiple dominates
the others, which can be deleted, so that the resulting matrix is
again triangular. Let $k$ be the location of the remaining column.
If Searcher searches $k$, then she also searches the locations from
$i$ up to $k$, so we can replace $X_k$ by $X_i\cdots X_k$. Similarly,
we can replace $X_{j+2}$ by $X_{k+1}\cdots X_{j+2}$. Essentially,
we have an OSR game on a total order, with a reduced number of locations.
By removing the runs, we end up with a game in which only the
final location possibly has expected value $>1$. The previous solution
of the OSR game with uncorrelated Bernoulli random variables carries
over to this situation.
\medbreak

\textbf{Random rescue sets.} Let $X$ denote the poset of locations, as
before. A draw from $\{0,1\}$-valued random variables~$\beta_x$
corresponds to a random subset $R\subset X$, which we call the
\emph{rescue set}. It is selected by Nature.
The players do not know which $R$ is drawn, but they do know Nature's probability
distribution on the family of subsets $2^X$.
For instance, weather conditions could limit the search to $k$ out of
$n$ locations, but it is impossible to predict which locations will
have bad weather.
In this example
Nature draws draws uniformly from the subsets of cardinality $k$. 
If $k=1$ then we have the special case of a CSR
game on an unordered $X$ in which every location has rescue probability
$\frac 1n$. If $k=n$ then all locations can be searched and the only
restriction is the order. This is an example of the hypergraph incidence game~\cite{fgt}.

In an SR game the search halts once it reaches a location
$x\not\in R$.
Bad weather may prevent a successful rescue in a certain location, 
but it does not stop the operation.
A search can continue even if it
visits a location that is not in $R$.
Let $S\subset X$ be the set of searched locations, i.e., the image of $\sigma$.
Then the payoff is one if $h\in R\cap S$ and zero otherwise.
This is an extension of the hypergraph incidence game, in which Searcher
chooses an edge $S$, Hider chooses a location $h$, and Searcher wins
if $h\in S$. 

\medbreak
\textbf{Search and Recovery.} 
In 2009 Air France Flight AF 447 
disappeared on its way from Rio to Paris over the middle of the Atlantic.
After several unsuccessful search operations, the wreckage was recovered
two years later using Bayesian search~\cite{stone}. In this approach,
the probability distribution of the location of the wreckage (the prior)
is updated after a search by a Bayesian factor that compensates for
the probability of detection. We can easily adapt SR games to this 
situation, by interpreting $\beta_x$ as the probability of detection
instead of the probability of successful search.
As in the case of rescue sets in the previous example, a search is allowed to continue if
$\beta_x$ is equal to zero.
This is Bayesian search against a worst-case prior distribution.
If Hider can move once search is over and before a new search begins,
then the game models a manhunt or a hunt for prey. Similar games have been
studied by Owen and McCormick~\cite{owen} and Gal and Casas~\cite{casas}.

\section{Conclusion}
We have presented a new type of Search and Rescue game on partially ordered
sets. We were able to solve the game in polynomial time in terms of the
number of locations for uncorrelated random variables.
We showed that the SR game on a tree can be interpreted as a game with
correlated random variables, for which the distribution is completely
reducible. 
This game was solved by Lidbetter under the restriction that 
Bayesian factors are bounded by one from above.
If they are bounded by one from below, the solution appears 
to be more difficult.
We conjecture that an optimal search involves a backjumping search.
A full solution of the game for general Bayesian factors seems to
be out of reach,
but we did find general bounds on the value of the game.
We hope that SR games on posets and Bayesian networks 
provide a fruitful avenue of
further research and applications.

\bibliographystyle{siam}
\bibliography{sandr}
\end{document}